\documentclass[article,plain]{llncs}
\usepackage{graphicx}
\usepackage[english]{babel}
\usepackage[utf8]{inputenc}
\usepackage{hyperref}
\usepackage{wrapfig}
\usepackage{stmaryrd}
\usepackage{amssymb}
\usepackage{amsmath}
\usepackage{amsfonts}
\usepackage{textcomp}
\usepackage[autostyle=true]{csquotes}
\usepackage{orcidlink}

\newcommand{\define}{\triangleq}
\newcommand{\net}{\mathit{N}}
\newcommand{\netn}{\mathcal{N}}
\newcommand\place{\mathsf{P}}
\newcommand\trans{\mathsf{T}}
\newcommand\placen{\mathsf{p}}
\newcommand\transn{\mathsf{t}}
\newcommand\pond{\mathsf{W}}
\newcommand{\postset}[1]{{{#1}{}^{\bullet}}}
\newcommand{\preset}[1]{{{}^{\bullet}\hspace*{-.5mm}{#1}}}
\newcommand\markn{\mathbf{m}}
\newcommand\enab{\mathsf{enab}}
\newcommand{\reach}{\mathbf{RS}}
\newcommand{\limreach}{\lim\!-\mathbf{RS}}
\newcommand\support[1]{[[#1]]}

\begin{document}

\title{Continuous Petri Nets Faithfully Fluidify \\ Most Permissive Boolean Networks}
\author{Stefan Haar\inst{1}\orcidlink{0000-0002-1892-2703} and Juri Kol\v{c}\'ak\inst{2}\orcidlink{0000-0002-9407-9682}} 
\institute{INRIA, France\\
\email{stefan.haar@inria.fr}\and 
Faculty of Technology, Bielefeld University, Germany\\\email{juri.kolcak@uni-bielefeld.de}}
\authorrunning{S. Haar  and J. Kol\v{c}\'ak}

\maketitle
\begin{abstract}
The analysis of biological networks has benefited from the richness of Boolean networks (BNs) and the associated theory.
These results have been further fortified in recent years by the emergence of \emph{Most Permissive (MP)} semantics, combining efficient analysis methods with a greater capacity of explaining pathways to states hitherto thought unreachable, owing to limitations of the classical update modes.
While MPBNs are understood to capture any behaviours that can be observed at a lower level of abstraction, all the way down to continuous refinements, the specifics and potential of the models and analysis, especially attractors, across the abstraction scale remain unexplored.
Here, we fluidify MPBNs by means of \emph{Continuous Petri nets (CPNs)}, a model of (uncountably infinite) dynamic systems that has been successfully explored for modelling and theoretical purposes.
CPNs create a formal link between MPBNs and their continuous dynamical refinements such as ODE models.
The benefits of CPNs extend beyond the model refinement, and constitute well established theory and analysis methods, recently augmented by abstract and symbolic reachability graphs.
These structures are shown to compact the possible behaviours of the system with focus on events which drive the choice of long-term behaviour in which the system eventually stabilises.
The current paper brings an important keystone to this novel methodology for biological networks, namely the proof that extant PN encoding of BNs instantiated as a CPN simulates the MP semantics.
In spite of the underlying dynamics being continuous, the analysis remains in the realm of discrete methods, constituting an extension of all previous work.
\end{abstract}

\section{Introduction}

Boolean networks (BNs) are a powerful and well-established formalism for modelling interacting entities as discrete dynamical systems, enjoying numerous and fruitful applications in live sciences, especially in the area of gene regulation~\cite{WangSA12,deJong02,Martinez-SosaM13,FaureNCT06}.

The choice of  schedule for  variable updates -- should variables be considered to update their values sequentially, or simultaneously, or in some mixture of those -- influences heavily the dynamics of the model.
BNs are known and used with multiple different semantics (updating modes), typically:
\begin{itemize}
    \item fully asynchronous (often simply asynchronous): only one variable, chosen nondeterministically, updates value in an update round;
    \item synchronous: all variables update value simultaneously; and 
    \item (generalised) asynchronous: any nondeterministically chosen subset of variables update value simultaneously.
\end{itemize}

(Generalised) synchronous semantics had long been considered the all-encompassing updating mode of BNs, capable of expressing all possible behaviours.
In~\cite{ManganA03}, however, the authors introduce an element of gene regulation which exhibits behaviours that cannot be explained with asynchronous BNs.
The considered system can be thought of as a \enquote{pulse generator}, visualised below be means of interactions as well as the approximate experimental expression profile of specie $3$.

\begin{center}
    \includegraphics{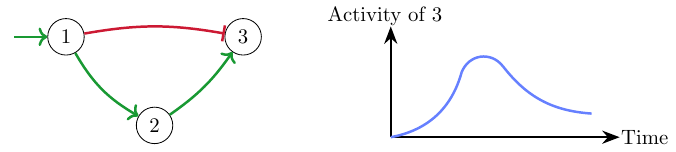}
\end{center}

An incoming signal activating $1$ causes a short term spike in activity of $3$, thus generating a pulse.
In a continuous dynamical model, such as ODE, this behaviour can be modelled using precise timing parameters. The positive influences from $1$ to $2$ and $2$ to $3$ are fast-acting, allowing the activation of $3$, while the negative influence from $1$ to $3$ is slow-acting, only deactivating $3$ after a pulse has been generated.

This behaviour can be equivalently thought of in the terms of activation thresholds.
The positive influences would have low activation thresholds, thus $2$ and subsequently $3$ start activating already at small increase of $1$.
On the other hand, the negative influence inhibiting $3$ has a high activation threshold, and is only enabled once $1$ reaches a much higher concentration.

The abstraction of such activation thresholds implicit in the asynchronous semantics is too coarse-grained to express the desired behaviour.
The only behaviours possible with asynchronous BNs are those  where $3$ never activates due to the inhibition from $1$, or $2$ activates $3$ in spite of that inhibition, resulting in a continuous activation, rather than a pulse.

Most permissive (MP) semantics of BNs have been developed~\cite{PauleveKCH20,ChatainHP18,ChatainHKPT20} to address this challenge.
Instead of updating variable value from $0$ to $1$ and vice-versa directly, MP semantics introduces intermediate dynamic states which can be read, nondeterministically, as either $0$ or $1$ in the outgoing interactions, thus effectively allowing for any combination of activation thresholds.
MPBNs can indeed not only explain the pulse generator at the level of the Boolean abstraction, but are proven to capture all behaviours expressed at a level of any multivalued or continuous refinement~\cite{PauleveKCH20}.

MPBNs in opposition to ODEs are the extremal points of the abstraction scale of dynamical system models.
The exact nature of the space between, threading between discrete and continuous models is, however, largely unexplored.

In this work we hope to open avenues for the study of that middle ground by establishing a connection between MPBNs and continuous Petri nets (CPNs).
CPNs combine the discrete and continuous worlds in terms of discrete time and continuous state space.
Moreover, CPNs are a good candidate for this connection due to the existing links to both ODEs and BNs.
By imposing timing constraints by the means of timed CPNs~\cite{DavidA10}, it becomes possible to fluidify time down to the continuous level, so that CPNs approximate ODE models.
On the other hand, by considering discrete state space, one lands in the well studied domain of discrete Petri nets (DPNs), often referred to simply as Petri nets.

The connection between DPNs and Boolean, or more generally, multivalued discrete networks, and especially encoding of said networks as DPNs, is a well studied topic~\cite{ChaouiyaRT06,ChaouiyaNRT11,ChatainHJPS14,ChatainHKPT20}.
in this work we demonstrate that the encoding of fully asynchronous BNs as safe DPNs used in those works, is also pertinent for encoding MPBNs as CPNs with no additional modifications.

\section{Definitions}
\subsection{Boolean Networks}

A BN of dimension $n\in\mathbb{N}$ (on $n$ variables) is a Boolean function $f\colon \mathbb{B}^n \rightarrow \mathbb{B}^n$, typically treated component-wise as local functions $f_i\colon \mathbb{B}^n \rightarrow \mathbb{B}$ giving the evolution of the variables $i\in\{1,\dots,n\}$, $f=(f_1,\dots, f_n)$.
Assigning a value to each variable constitutes a \emph{configuration} $\mathbf{x}\colon\{1,\dots, n\} \rightarrow\mathbb{B}$, equivalently defined as an $n$-ary Boolean vector, $\mathbf{x}\in\mathbb{B}^n$.

\begin{figure}
    \centering
    \vskip-20pt
    \begin{minipage}{0.25\textwidth}
        \vspace*{5mm}
        $f_1(\mathbf{x})= \neg \mathbf{x}_2$
       
        $f_2(\mathbf{x})= \neg \mathbf{x}_1$
       
        $f_3(\mathbf{x})= \neg \mathbf{x}_1 \land \mathbf{x}_2$.
     \vspace*{5mm}
    \end{minipage}
    \hspace*{0.2\textwidth}
    \begin{minipage}{0.4\textwidth} 
        \vspace*{4mm}
        \includegraphics{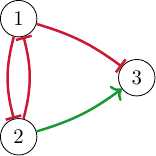}
        \vspace*{4mm}
    \end{minipage}
    \vskip-10pt
    \caption{A running example BN on three variables (left) with its interaction graph (right). Positive interactions are depicted by pointed green arrows while negative interactions are depicted by flat red arrows.}
    \label{fig:running_example}
\end{figure}

Boolean networks are typically visualised by their abstraction in the form of an \emph{interaction graph} (influence graph).
An example BN network defined by local functions and the associated interaction graph are given in Figure~\ref{fig:running_example}.

An interaction graph links variables together based on the interactions between them and their nature (positive vs negative), e.g. the mutual negative interactions between variables $1$ and $2$ in the example BN.
The exact interplay of said interactions, e.g. the \enquote{and} in $f_3$ as opposed to an \enquote{or}, are abstracted away.

For a BN $f$ of dimension $n$ and any two configurations $\mathbf{x},\mathbf{y}\in \mathbb{B}^n$, let $\Delta(\mathbf{x},\mathbf{y})=\{i\in\{1,\dots,n\}\mid \mathbf{x}_i\neq \mathbf{y}_i\}$ be the set of all variables whose value differs between $\mathbf{x}$ and $\mathbf{y}$.
Then the fully asynchronous semantics of $f$, $fa(f)$, synchronous semantics of $f$, $syn(f)$, and asynchronous semantics of $f$, $ga(f)$ are defined as follows for any pair of configurations $\mathbf{x},\mathbf{y}\in\mathbb{B}^n$:
\begin{align*}
    (\mathbf{x},\mathbf{y})\in fa(f)&\Longleftrightarrow\exists i\in\{1,\dots,n\}, \Delta(\mathbf{x},\mathbf{y})=\{i\} \wedge f_i(\mathbf{x})=y_i \\
    (\mathbf{x},\mathbf{y})\in syn(f)&\Longleftrightarrow\forall i\in\{1,\dots,n\},f_i(\mathbf{x})=y_i \\
    (\mathbf{x},\mathbf{y})\in ga(f)&\Longleftrightarrow\forall i\in\Delta(\mathbf{x},\mathbf{y}),f_i(\mathbf{x})=y_i
\end{align*}
It follows immediately that $fa(f)\subseteq ga(f)$ as well as $syn(f) \subseteq ga(f)$.

We use the natural infix notation for the semantics relations, e.g. $\mathbf{x}\xrightarrow[f]{fa} \mathbf{y}$.
We further omit the BN $f$ when its obvious from context, $\mathbf{x}\xrightarrow{fa} \mathbf{y}$.
The associated reachability relations are obtained as the reflexive and transitive closures of the semantics relations, $\mathbf{x}\underset{f}{\overset{fa}{\longrightarrow^*}} \mathbf{y}$, etc.

Configurations of the example BN from Figure~\ref{fig:running_example} reachable from $(0,0,0)$ under fully asynchronous or asynchronous semantics are visualised in Figure~\ref{fig:fa_ts}.
Synchronous semantics are omitted, as the only reachable configurations are $(1,1,0)$ and $(0,0,0)$ itself (the cycle in asynchronous reachability).

\begin{figure}
    \centering
    \includegraphics{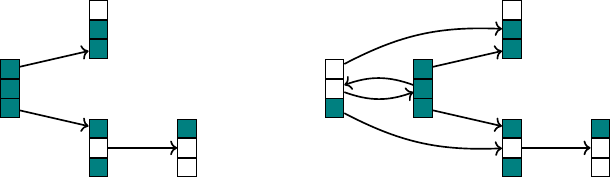}
    \caption{All configurations of the example BN (Figure~\ref{fig:running_example}) reachable from $(0,0,0)$ under fully asynchronous semantics (left) and asynchronous semantics (right). The variables are ordered from top to bottom and dark tiles represent value $0$ while light tiles represent value $1$. Self-loops are omitted.}
    \label{fig:fa_ts}
\end{figure}

Finally, we introduce the MP semantics of BNs.
Different formalisations of the MP semantics have been considered~\cite{PauleveKCH20,PauleveS21}.
To remain conceptually close to state updates of PNs, we recall the original definition~\cite{PauleveKCH20}.
MP differ from other semantics by introduction of two transient (or dynamic) values, $\nearrow$ to signify that a variable is increasing from $0$ to $1$ and $\searrow$ to signify that a variable is decreasing from $1$ to $0$.

For a BN $f$ of dimension $n$ and any MP configuration $\hat{\mathbf{x}}\in\{0,1,\nearrow,\searrow\}^n$, let $\beta(\hat{\mathbf{x}}) = \{\mathbf{x}\in\mathbb{B}^n\mid \forall i\in\{1,\dots,n\}, \hat{\mathbf{x}}_i\in\mathbb{B} \Longrightarrow \mathbf{x}_i=\hat{\mathbf{x}}_i\}$ be the set of compatible binarisations expressed as Boolean configurations.
Then, the MP semantics of a BN $f$, $\hat{mp}(f)$, is defined as follows for any pair of MP configurations $\hat{\mathbf{x}},\hat{\mathbf{y}}\in\{0,1,\nearrow,\searrow\}^n$:
\begin{align*}
    (\hat{\mathbf{x}},\hat{\mathbf{y}})\in \hat{mp}(f)\Longleftrightarrow \exists i\in\{1,\dots,n\}, \Delta(\hat{\mathbf{x}},\hat{\mathbf{y}})=\{i\} \wedge [(\hat{\mathbf{x}}_i=\nearrow \wedge \hat{\mathbf{y}}_i=1)& \vee \\
    (\hat{\mathbf{x}}_i=\searrow \wedge \hat{\mathbf{y}}_i=0)& \vee \\
    (\hat{\mathbf{x}}_i\neq 1\wedge \hat{\mathbf{y}}_i=\nearrow \wedge \exists \mathbf{x}\in\beta(\hat{\mathbf{x}}), f_i(\mathbf{x})=1)& \vee \\
    (\hat{\mathbf{x}}_i\neq 0\wedge \hat{\mathbf{y}}_i=\searrow \wedge \exists \mathbf{x}\in\beta(\hat{\mathbf{x}}), f_i(\mathbf{x})=0)&]
\end{align*}
The top two lines of the disjunction enable a variable in a transient value to collapse into the target Boolean value at any time.
The bottom two lines dictate that for a variable to change to a transient value, a binarisation must exist for which the local function returns the target Boolean value of said transient value.

As we are mostly interested in Boolean configurations, we restrict the MP semantics accordingly, thus for any $\mathbf{x},\mathbf{y}\in\mathbb{B}^n$:
\[
    (\mathbf{x},\mathbf{y}) \in mp(f) \Longleftrightarrow \mathbf{x} \overset{\hat{mp}}{\underset{f}{\longrightarrow^*}}\mathbf{y}
\]

We have $ga(f) \subseteq mp(f)$ as any asynchronous update is reproduced by first putting all relevant variables into transient value and subsequently collapsing them.
The inclusion is strict in general, as shown in Figure~\ref{fig:mp_ts} of all Boolean configurations of the example BN from Figure~\ref{fig:running_example} reachable under MP semantics.

\begin{figure}[th]
    \centering
    \includegraphics{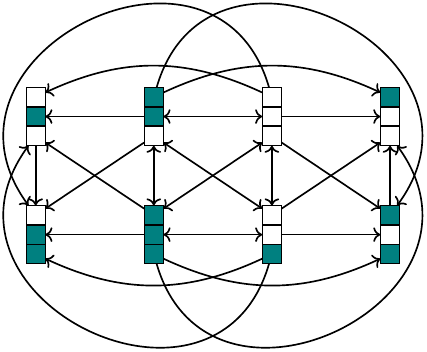}
    \caption{All Boolean configurations of the example BN (Figure~\ref{fig:running_example}) reachable from $(0,0,0)$ under MP semantics. While the fixed points are preserved, all other configurations become transiently reachable, as opposed to asynchronous semantics.}
    \label{fig:mp_ts}
\end{figure}

\subsection{Petri Nets}

\emph{Petri Nets (PNs)} have as skeleton bipartite directed graphs with two sets of nodes, \emph{places} and \emph{transitions}.
Places model resource types available in the system, and can hold a certain amount of \emph{tokens} or \emph{token mass} representing availability of said resource.
Transitions allow removing (reducing) or adding tokens (token mass), according to the \emph{arcs} connecting them to places.

Formally, a \emph{net} is a tuple $\net=(\place,\trans,\pond)$, where $\place$ is a finite set of places, $\trans$ is a finite set of transitions and $\pond\colon (\place\times \trans)\cup(\trans\times \place)\rightarrow \mathbb{N}_0$ is a weight function assigning each transition the amount of resources consumed or produced in each place.
For $x,y\in \place\cup\trans$, we say there is an arc from $x$ to $y$ if $\pond(x,y) > 0$.
A simple example net is visualised in Figure~\ref{fig:re_pn}.

\begin{figure}
    \centering
    \includegraphics{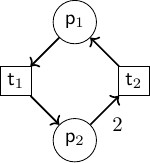} \hspace*{12mm} \includegraphics{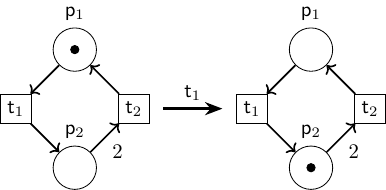}
    \caption{Left: An example net with two places (depicted as circles) and two transitions (depicted as squares). The arc weights are given as labels on the arrows. As per convention, arc weights equal to $1$ are omitted.
    Right: An example of a transition firing in a DPN with the same net structure. The tokens assigned by the markings are depicted as black dots inside the places.}
    \label{fig:re_pn}  \label{fig:dpn_reach}
\end{figure}

For each place or transition $x\in \place\cup \trans$, we call $\preset{x} = \{y\in\place\cup\trans\mid \pond(y,x)>0\}$ the \emph{pre-set} of $x$.
Similarly, $\postset{x} = \{y\in\place\cup\trans\mid \pond(x,y)>0\}$ is the \emph{post-set} of $x$.

\subsubsection{Discrete Petri nets.}
A DPN $\netn = (\net,M_0)$ is a net $\net=(\place,\trans,\pond)$
with a (discrete) \emph{marking} $M_0\colon \place\rightarrow\mathbb{N}_0$, called the initial marking.
A discrete marking models the resource allocation in the system by assigning a natural number of tokens to each place.

A transition $\transn \in\trans$ is \emph{enabled} at a marking $M$, denoted $M\xrightarrow{\transn}$, if and only if $\forall \placen\in\preset{\transn}, M(\placen)\geq \pond(\placen,\transn)$.
An enabled transition can fire, $M\xrightarrow{\transn}M'$, leading to a new marking $M'$ defined for each $\placen\in\place$ as follows:
\[
    M'(\placen) \define M(\placen) - \pond(\placen, \transn) + \pond(\transn,\placen)
\]

Let us consider the example net from Figure~\ref{fig:re_pn} with the initial marking $M_0=\{(\placen_1, 1),(\placen_2, 0)\}$.
Transition $\transn_1$ is enabled at $M_0$, leading to marking $M_1=\{(\placen_1, 0),(\placen_2, 1)\}$ as illustrated in Figure~\ref{fig:dpn_reach} on the right hand side.
$M_1$, however, is a \emph{deadlock}, a marking at which no transition is enabled.

A \emph{finite firing sequence} is a word $w=(\transn_i)_{i< n\in\mathbb{N}_0}$ over $\trans$ for which there exists a sequence of markings $(M_i)_{i< n+1}$ such that $M_i\xrightarrow{\transn_i}M_{i+1}$ for all $i< n$.
We denote the set of markings reachable from some marking $M$ in $\netn$ by $\reach_\netn(M) = \{M'\mid\exists w\in\trans^*, M\xrightarrow{w}M'\}$, or simply $\reach(M)$ where $\netn$ is obvious.

A DPN all of whose reachable markings have at most one token per place, $\forall M\in \reach(M_0), \forall \placen\in\place, M(\placen)\leq 1$, is \emph{safe} (also $1$-bounded).
The PN encoding of BNs used in this paper produces safe DPNs, which explains our interest.
By abuse of notation,  markings of safe DPNs can be seen  as sets, $M=\{\placen\in\place\mid M(\placen)=1\}$.

\subsubsection{Continuous Petri nets.}
A CPN $\netn = (\net,\markn_0)$ is a net $\net=(\place,\trans,\pond)$  with a (continuous) marking $\markn_0\colon \place\to\mathbb{R}_0^+$, called the initial marking.
Similarly to the discrete case, continuous markings also model resource allocation, but assign $\mathbb{R}_0^+$-valued token mass instead of discrete tokens. 
A transition $\transn \in\trans$ has an \emph{enabling degree} at marking $\markn$ given by $\enab(\transn,\markn) = \min_{\placen\in\preset{\transn}}\frac{\markn(\placen)}{\pond(\placen,\transn)}$ when $\preset{\transn}\ne\emptyset$, and equal to $\infty$ otherwise.
A transition with $\enab(\transn,\markn) > 0$ is enabled at $\markn$ and can $\alpha$-fire for any $\alpha \in [0,\enab(\transn,\markn)]\cap \\mathbb{R}$, $\markn\xrightarrow{\alpha\transn}\markn'$, leading to a new marking $\markn'$ defined for each $\placen\in\place$ as follows:
\[
    \markn'(\placen)\define\markn(\placen) - \alpha\pond(\placen, \transn) + \alpha\pond(\transn,\placen)
\]

A finite firing sequence is a finite word $\sigma=(\alpha_i \transn_i)_{i< n\in\mathbb{N}_0}$ over $\mathbb{R}_0^+\times\trans$ for which there exists a finite sequence of markings $(\markn_i)_{i< n+1}$ such that $\markn_i\xrightarrow{\alpha_i\transn_i}\markn_{i+1}$ for all $i< n$.
An infinite firing sequence is an infinite word $\sigma=(\alpha_i t_i)_{n\in\mathbb{N}_0}$ over $\mathbb{R}_0^+\times\trans$ for which there exists an infinite family of markings $(\markn_i)_{i\leq \omega}$ such that $\markn_i\xrightarrow{\alpha_i\transn_i}\markn_{i+1}$ for all $i\in\mathbb{N}_0$ and $\lim_{i\rightarrow \infty} \markn_{i}=\markn_{\omega}$.

We denote the set of markings reachable from some marking $\markn$ in $\netn$ by $\reach_\netn(\markn) = \{\markn'\mid\exists \sigma\in(\mathbb{R}_0^+\times\trans)^*, \markn\xrightarrow{\sigma}\markn'\}$.
We further denote the set of markings reachable in the limit (lim-reachable) from some marking $\markn$ in $\netn$ by $\limreach_\netn(\markn) = \{\markn'\mid\exists \sigma\in(\mathbb{R}_0^+\times\trans)^\omega, \markn\xrightarrow{\sigma}\markn'\}$.
Since an infinite firing sequence can include null firings ($\alpha=0$), we have $\reach_\netn(\markn) \subseteq \limreach_\netn(\markn)$.
We write simply $\reach(\markn)$ and $\limreach(\markn)$ where $\netn$ is obvious.

Similar to DPNs, a CPN all of whose reachable markings have token mass at most $1$ per place, $\forall \markn\in \reach(\markn_0), \forall \placen\in\place, \markn(\placen)\leq 1$, is safe. The notions of safeness and lim-safeness coincide~\cite{RecaldeTS99}.

The ability to fire a transition by varying degree means that there are infinitely many possible firings in any marking at which at least one transition is enabled.
We therefore cannot visualise all the possible transition firings of a CPN in the same fashion as Figure~\ref{fig:dpn_reach} does for a DPN.
Symbolic analysis methods have therefore been introduced  for CPN dynamics, namely abstract reachability graphs (ARG) and symbolic reachability trees (SRT)~\cite{HaarH24}.

The ARG constitutes an abstraction of the state space by collapsing markings that have the same \emph{support}, i.e. the set of places which have non-zero token mass, $\support{\markn} \define \{\placen\in\place\mid \markn(\placen)>0\}$, into equivalence classes.
Supports are linked between one another by two types of edges, \emph{anonymous} and \emph{border}, both representing the possibility of the system evolving into a marking with different support.
Border edges are of special interest, as they represent occurrences that \emph{permanently disable} one or more transitions, thus shaping the reachability landscape.

The SRT provides a refined view at such a landscape in the form of a directed tree linked by border edges.
The nodes of the tree are aggregations of markings by their \emph{mode}, i.e.  the set of transitions eventually fireable from the given marking: $T_\markn \define \{\transn\in\trans\mid \exists \markn\xrightarrow{\sigma}, \transn\in\support{\sigma} \}$, where $\support{\sigma}$ is the set of all transitions occurring along $\sigma$ with $\alpha > 0$.
The branches of the tree then indicate refinements of those modes, as less and less transitions remain available.
The depth of the SRT is naturally bounded by $|\trans|$.

We omit the formal definition of ARG and SRT due to space constraints.
The reader is instead referred to~\cite{HaarH24}.
Both the ARG and SRT of the example net are illustrated in Figure~\ref{fig:re_arg}.

\begin{figure}
    \centering
    \includegraphics[width=\textwidth]{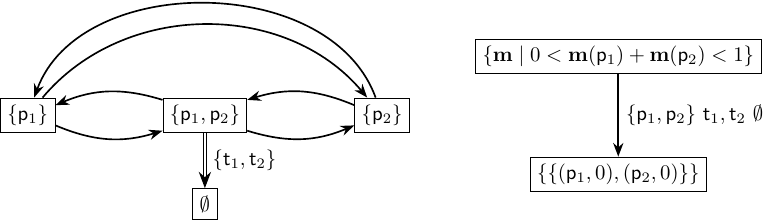}
    \caption{The ARG (left) and the SRT (right) of the example net from Figure~\ref{fig:re_pn}. The border edges are given as double line arrows in the ARG. The SRT for this example is trivial, as there is only one border edge.}
    \label{fig:re_arg}
\end{figure}

Observe that as opposed to the DPN (Figure~\ref{fig:dpn_reach}), $\transn_2$ can be fired in a CPN with the same initial marking $\markn_0=M_0$.
Firing of $\transn_2$ halfs the token mass, every \enquote{loop} around the net thus results in a marking with less total token mass.
This phenomenon is illustrative, as the two places of the net are a trap, that is $\postset{\{\placen_1,\placen_2\}} \subseteq \preset{\{\placen_1,\placen_2\}}$.
It is a well known result for DPNs that an initially marked trap cannot be emptied~\cite{DeselE95}.
This result extends to reachability in CPNs, indeed after any finite number of transition firings, there will always be some token mass left, however, the empty marking is lim-reachable.

\section{Petri Net Encoding of Boolean Networks}

In this section, we introduce the encoding of BNs as safe PNs. We make use of an elegant encoding given in~\cite{ChatainHKPT20}. Due to the focus on concurrency, the original encoding contains so called contextual (or read) arcs. Read arcs can be replaced equivalently with loops, allowing to obtain the same behaviour with an ordinary net structure. We reiterate the definition of~\cite{ChatainHKPT20} without read arcs, but we first require some additional notation for formal reasoning about Boolean functions.

For any variable $i\in\{1,\dots,n\}$ of a BN $f$, $\delta^+(i) = DNF(\neg x_i \wedge f_i(x))$ denotes the set of conjunctive clauses of the disjunctive normal form of the function allowing $i$ to increase value.
Similarly, $\delta^-(i) = DNF(x_i \wedge \neg f_i(x))$ the set of clauses of the function allowing $i$ to decrease value.
Finally, $\delta(i) = \delta^+(i) \cup \delta^-(i)$ is the set of all clauses pertaining to variable $i$.

We say a configuration $\mathbf{y}\in\mathbb{B}^n$ satisfies a clause $C$, $\mathbf{y}\models C$, if for each positive literal $[x_j]\in C$, $\mathbf{y}_j = 1$ and for each negative literal $[\neg x_j]\in C$, $\mathbf{y}_j = 0$.

\begin{definition}[Petri Net Encoding of Boolean Networks]
Given a BN $f$ of dimension $n$, the Petri net encoding of $f$ is the net $\net(f)=(\place,\trans,\pond)$ where $\place \define \{1,\dots, 2n\}$, $\trans\define \bigcup_{i\in \{1,\dots, n\}} \{\transn_C\mid C \in \delta(i)\}$ and $\pond$ is $1$ in all the following cases for every $i\in\{1,\dots,n\}$, and $\pond(x,y)\define0$ otherwise:
\begin{eqnarray}
\label{enc:increase} \forall C\in \delta^+(i), \pond(i,\transn_C)=\pond(\transn_C,i+n)&\define&1 \\
\label{enc:decrease} \forall C\in \delta^-(i), \pond(i+n,\transn_C)=\pond(\transn_C,i)&\define&1 \\
\label{enc:inhibitors} \forall C\in\delta(i), \forall [\neg x_j] \in C, j\neq i, \pond(j,\transn_C)=\pond(\transn_C,j)&\define&1 \\
\label{enc:activators} \forall C\in\delta(i), \forall [x_j] \in C, j\neq i, \pond(j+n,\transn_C)=\pond(\transn_C,j+n)&\define&1
\end{eqnarray}
\label{def:pn_encoding}
\end{definition}
Each variable $i\in\{1,\dots, n\}$ is represented by two places, $i, i+n$, which correspond to variable $i$ valued $0$ and $i$ valued $1$.
The transitions correspond to the clauses of the local functions of the BN.
For each variable $i\in\{1,\dots, n\}$, we use $T^+_i = \{t_C\in T|C\in \delta^+(i)\}$ to denote the transitions that increase the value of $i$. Similarly, $T^-_i = \{t_C\in T|C\in \delta^-(i)\}$ denotes the transitions decreasing value of $i$, and $T_i=T^+_i\cup T^-_i$ are all transitions acting on the variable $i$.

Note that $\net(f)$ is up to exponentially larger than  $f$; however, as the exponent depends only on the maximal in-degree  in the interaction graph, rather than the vertex count, the blow-up is manageable for sparse networks.

Finally, the arcs fall into two categories. The first two lines (\ref{enc:increase},\ref{enc:decrease}) describe variables increasing, respectively decreasing, value as an effect of the associated transition.
The last two lines (\ref{enc:inhibitors},\ref{enc:activators}) describe the loops on places which encode the negative and positive interactions with the target variable.

The above structure guarantees two useful properties.
For each variable $i\in\{1,\dots,n\}$, only transitions in $T_i$ can modify the marking of the places $i, i+n$, and $\{i, i+n\}$ is a P-invariant, meaning that the sum of the tokens (token mass) in $i$ and $i+n$ remains constant across all reachable markings.

Let $\mathbf{x} \in \mathbb{B}^n$ be a configuration of the BN $f$. We define the corresponding discrete marking of $\net(f)$ as $\langle\mathbf{x}\rangle = \left\{i+n\mathbf{x}_i \mid i\in \{1,\dots, n\}\right\}$.
The net encoding of the example BN with initial marking $\langle(0,0,0)\rangle=\{1,2,3\}$ is depicted in Figure~\ref{fig:pn_encoding}.

\begin{figure}
    \centering
    \includegraphics[width=\textwidth]{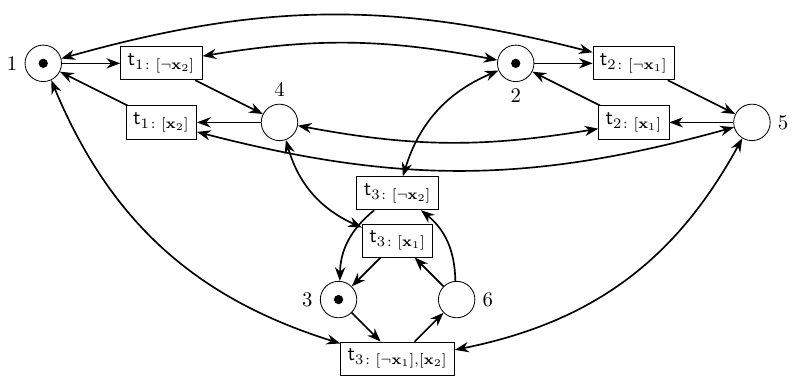}
    \caption{The PN encoding of the example BN given in Figure~\ref{fig:running_example}. The loops are depicted as double headed arrows to avoid clutter.}
    \label{fig:pn_encoding}
\end{figure}

As a DPN, either $\transn_{1\colon[\neg \mathbf{x}_2]}$ to activate $1$ or $\transn_{2\colon[\neg \mathbf{x}_1]}$ to activate 2 can be fired, reaching $\{2,3,4\} = \langle(1,0,0)\rangle$ or $\{1,3,5\}=\langle(0,1,0)\rangle$.
The first case is a deadlock, corresponding to the fixed point of the BN.
In the second case, $\transn_{3\colon[\neg\mathbf{x}-1],[\mathbf{x}_2]}$ can be fired to reach the second deadlock $\{1,5,6\} = \langle(0,1,1)\rangle$, perfectly reproducing the fully asynchronous semantics of the BN (Figure~\ref{fig:fa_ts}).

Observe that the synchronous update reaching configuration $(1,1,0)$ in the asynchronous semantics is not captured by a DPN.
This follows from the fact that $\{1,2\}\subseteq\place$ is an initially marked trap, and therefore cannot be emptied.
Synchronous or asynchronous semantics of BNs thus cannot be reproduced by a DPN without further structural modification or special semantics~\cite{ChatainHKPT20}.
We show that such synchronous updates, as well as any other updates possible under MPBNs are recovered by CPNs.

Theorem 1 of~\cite{ChatainHKPT20} proves that the DPN $\net(f)$ emulates $f$ with fully asynchronous semantics, i.e. any update of the BN is reflected by an enabled transition between the corresponding markings of the DPN and vice-versa, formally:
\[
    \forall \mathbf{x},\mathbf{y} \in \mathbb{B}^n, \mathbf{x} \xrightarrow[f]{fa} \mathbf{y} \Longleftrightarrow \exists \transn\in\trans, \langle \mathbf{x}\rangle \xrightarrow[\net(f)]{\transn} \langle \mathbf{y}\rangle
\]
This equivalence naturally extends to the reachability relation:
\[
    \forall \mathbf{x},\mathbf{y} \in \mathbb{B}^n, \mathbf{x} \underset{f}{\overset{fa}{\longrightarrow^{*}}} \mathbf{y} \Longleftrightarrow \langle \mathbf{y}\rangle \in\reach_{\net(f)}(\langle \mathbf{x}\rangle)
\]

In the following, we show that with the same PN encoding of Definition~\ref{def:pn_encoding} instantiated as a CPN, a reachability equivalence can be established between a BN with the MP semantics and the CPN encoding.
We first establish additional notation.

Since the places of the net encoding are an interval of natural numbers, we treat the continuous markings as vectors, $\markn\in{\mathbb{R}_0^+}^{2n}$.
Let now $\hat{\mathbf{x}} \in \{0,1,\nearrow,\searrow\}^n$ be a configuration of the MPBN $f$.
We define the set of all compatible continuous markings as:
\begin{align*}
\mu(\hat{\mathbf{x}}) \define \{\markn \in {\mathbb{R}_0^+}^{2n}\mid \forall i\in\{1,\dots, n\}&, \markn_i+\markn_{i+n} = 1 \\
&\wedge \hat{\mathbf{x}}_i\in\mathbb{B}\Longrightarrow \markn_i=1-\hat{\mathbf{x}}_i \\
&\wedge \hat{\mathbf{x}}_i\in\{\nearrow,\searrow\}\Longrightarrow 0<\markn_i<1\}
\end{align*}

Of special interest is the marking $\langle\hat{\mathbf{x}}\rangle$ defined as follows for all $i\in\{1,\dots, n\}$:
\[
    {\langle \hat{\mathbf{x}}\rangle}_i, {\langle\hat{\mathbf{x}}\rangle}_{i+n}\define
	\begin{cases}
		1-\hat{\mathbf{x}}_i,\hat{\mathbf{x}}_i &\text{ if } \hat{\mathbf{x}}_i\in\mathbb{B}\\
		0.5,0.5 &\text{ if } \hat{\mathbf{x}}_i\in\{\nearrow,\searrow\}
	\end{cases}
\]

Clearly, $\langle\hat{\mathbf{x}}\rangle\in\mu(\hat{\mathbf{x}})$; for $\hat{\mathbf{x}}\in\mathbb{B}^n$, the definition coincides with the one for the discrete markings, $\langle\hat{\mathbf{x}}\rangle=\langle \mathbf{x}\rangle$ where $\{\mathbf{x}\} = \beta(\hat{\mathbf{x}})$.

\begin{lemma}
Given a BN $f$ of dimension $n$, a configuration $\hat{\mathbf{x}}\in\{0,1,\nearrow,\searrow\}^n$, a transition $\transn\in T_i$ and a marking $\markn\in\mu(\hat{\mathbf{x}})$ such that $\enab(\transn, \markn) > 0$, then for any other marking $\markn'\in\mu(\hat{\mathbf{x}})$, $\enab(\transn, \markn') > 0$.

If $\markn'=\langle\hat{\mathbf{x}}\rangle$, additionally $\enab(\transn, \langle\hat{\mathbf{x}}\rangle) \geq 0.5$.
\label{lem:mark_equivalence}
\end{lemma}

\begin{proof}
For any two markings $\markn, \markn'\in\mu(\hat{\mathbf{x}})$ and for each place $\placen\in\place$, it holds that $\markn_{\placen} = 0 \Longleftrightarrow \markn'_{\placen} = 0$. Since $\placen\in\preset{\transn}$ implies $\markn_{\placen}>0$, necessarily also $\markn'_{\placen} > 0$.

For $\langle\hat{\mathbf{x}}\rangle$, we have ${\langle\hat{\mathbf{x}}\rangle}_{\placen}\neq 0 \Longrightarrow {\langle\hat{\mathbf{x}}\rangle}_{\placen} \geq 0.5$ by definition.
\end{proof}

\begin{lemma}
Given a BN $f$ of dimension $n$, a transition $\transn\in T^+_i$ ($\transn\in T^-_i$) and a marking $\markn\in{\mathbb{R}_0^+}^{2n}$ such that $\enab(\transn,\markn) > 0$, then for any $\hat{\mathbf{x}}\in\{0,1,\nearrow,\searrow\}^n$ such that $\markn\in\mu(\hat{\mathbf{x}})$ and $\hat{\mathbf{x}}_i\in\{0,\searrow\}$ ($\hat{\mathbf{x}}_i\in\{1,\nearrow\}$), $\hat{\mathbf{x}}\xrightarrow[f]{\hat{mp}}\hat{\mathbf{y}}$ where $\hat{\mathbf{y}}_i=\nearrow$ ($\hat{\mathbf{y}}_i=\searrow$).
\label{lem:enab_transfers_to_bn}
\end{lemma}

\begin{proof}
Since the cases for $\transn\in T_i^+$ and $\transn\in T_i^-$ are symmetrical, we only show the proof for $\transn\in T_i^+$.
Let  $C(\transn)\in\delta^+(i)$ be the clause corresponding to $\transn$.

Let $[\neg \mathbf{x}_j]\in C(\transn)$ be an arbitrary negative literal. Since $\enab(\transn,\markn) > 0$, $\markn_j > 0$ and by $\markn\in\mu(\hat{\mathbf{x}})$, $\hat{\mathbf{x}}_j\in\{0,\nearrow,\searrow\}$.
Similarly, let $[\mathbf{x}_j]\in C(\transn)$ be an arbitrary positive literal. Again from $\enab(\transn,\markn) > 0$, $\markn_{i+n} > 0$ and by $\markn\in\mu(\hat{\mathbf{x}})$, $\hat{\mathbf{x}}_j\in\{1,\nearrow,\searrow\}$.

As such, there must exist a configuration $\mathbf{z} \in \beta(\hat{\mathbf{x}})$ such that $\mathbf{z}\models C(\transn)$, and thus variable $i$ is allowed to increase value, $\hat{\mathbf{x}}\xrightarrow[f]{\hat{mp}}\hat{\mathbf{y}}$.
\end{proof}

\begin{lemma}
Given a BN $f$ of dimension $n$ and two configurations $\hat{\mathbf{x}},\hat{\mathbf{y}}\in\{0,1,\nearrow,\searrow\}^n$ such that for any $i\in\{1,\dots,n\}$, ${\hat{\mathbf{x}}}_i\in \mathbb{B} \Longrightarrow {\hat{\mathbf{y}}}_i={\hat{\mathbf{x}}}_i$, then for any transition $\transn \in \trans$, $\enab(\transn,\langle\hat{\mathbf{y}}\rangle) \geq 0.5 \Longrightarrow \enab(\transn, \langle\hat{\mathbf{x}}\rangle) \geq 0.5$.
\label{lem:enab_monotony}
\end{lemma}

\begin{proof}
Let $\placen \in \preset{\transn}$ be arbitrary and let $j = \placen$, respectivelly $j = \placen - n$ for $\placen>n$, be the BN variable whose value assignment $\placen$ represents.
We now conduct discussion on the nature of the value of $j$ in $\hat{\mathbf{x}}$.

The case of $\hat{\mathbf{x}}_j\in\{\nearrow,\searrow\}$ is simple as we have ${\langle\hat{\mathbf{x}}\rangle}_\placen= 0.5$ by definition.

Let thus $\hat{\mathbf{x}}_j\in\mathbb{B}$.
By assumption, ${\hat{\mathbf{x}}}_j\in \mathbb{B} \Longrightarrow \hat{\mathbf{x}}_j=\hat{\mathbf{y}}_j$.
Since $\enab(\transn,\langle\hat{\mathbf{y}}\rangle) \geq 0.5$, we necessarily have ${\langle\hat{\mathbf{y}}\rangle}_\placen = 1 = {\langle\hat{\mathbf{x}}\rangle}_\placen$.
\end{proof}

\begin{lemma}
Given a BN $f$ of dimension $n$ and configurations $\hat{\mathbf{x}}\xrightarrow[f]{\hat{mp}} \hat{\mathbf{y}}$ such that $\hat{\mathbf{y}}_i = \nearrow$ ($\hat{\mathbf{y}}_i = \searrow$), where $i=\Delta(\hat{\mathbf{x}},\hat{\mathbf{y}})$, then there exists a transition $\transn\in T^+_i$ ($\transn\in T^-_i$) such that $\enab(\transn, \langle\hat{\mathbf{x}}\rangle) \geq 0.5$.
\label{lem:half_enabled}
\end{lemma}

\begin{proof}
Similarly to Lemma~\ref{lem:enab_transfers_to_bn}, the cases $\hat{\mathbf{y}}_i = \nearrow$ and $\hat{\mathbf{y}}_i = \searrow$ are symmetrical. We thus conduct the proof only for $\hat{\mathbf{y}}_i = \nearrow$.

Since $\hat{\mathbf{x}}\xrightarrow[f]{\hat{mp}} \hat{\mathbf{y}}$, there must exist a binarisation $\mathbf{z}\in \beta(\hat{\mathbf{x}})$ such that $f_i(\mathbf{z})=1$ and therefore a clause $C \in \delta^+(i)$, such that $\mathbf{z} \models C$.
Let $\transn_C \in T^+_i$ be the transition corresponding to $C$ as per Definition~\ref{def:pn_encoding}.

We know $\hat{\mathbf{x}}_i \in \{0, \searrow\}$ and thus ${\langle\hat{\mathbf{x}}\rangle}_i\geq 0.5$.
Let then $j\in\{1,\dots,n\}, j\neq i$ be an arbitrary variable such that $j\in\preset{\transn_C}$.
By Definition~\ref{def:pn_encoding}, $[\neg x_j] \in C$.
Then by $\mathbf{z}\models C$, $\mathbf{z}_j = 0$ and since $\mathbf{z}\in \beta(\hat{\mathbf{x}})$, $\hat{\mathbf{x}}_j\in\{0,\nearrow,\searrow\}$ and ${\langle\hat{\mathbf{x}}\rangle}_j \geq 0.5$.

Similarly, for $j\in\{1,\dots,n\}, j\neq i$ such that $j+n\in\preset{\transn_C}$, we have $[x_j]\in C$.
By $\mathbf{z}\models C$, $\mathbf{z}_j = 1$ and since $\mathbf{z}\in \beta(\hat{\mathbf{x}})$, $\hat{\mathbf{x}}_j\in\{1,\nearrow,\searrow\}$ and ${\langle\hat{\mathbf{x}}\rangle}_{j+n} \geq 0.5$.
\end{proof}

\begin{corollary}
Given a BN $f$ of dimension $n$ and configurations $\hat{\mathbf{x}}\xrightarrow[f]{\hat{mp}} \hat{\mathbf{y}}$ such that $\hat{\mathbf{x}}_i \in \mathbb{B}$, where $i=\Delta(\hat{\mathbf{x}},\hat{\mathbf{y}})$, then there exists $\transn\in T_i$ such that $\langle \hat{\mathbf{x}}\rangle \xrightarrow[\net(f)]{0.5\transn} \langle\hat{\mathbf{y}}\rangle$.
\label{cor:abstraction_one_to_one}
\end{corollary}

Corollary~\ref{cor:abstraction_one_to_one} is a direct application of Lemma~\ref{lem:half_enabled}: a value change from Boolean to transient in the MPBN is  directly reproducible in the CPN encoding.

\begin{theorem}[Continuous Petri Nets Simulate MP Boolean Networks]
Let $f$ be a Boolean network of dimension $n$.
Then, for any configurations $\mathbf{x},\mathbf{y} \in \mathbb{B}^n$, $\mathbf{x} \underset{f}{\overset{mp}{\longrightarrow^{*}}} \mathbf{y} \Longleftrightarrow \langle \mathbf{y}\rangle\in \limreach_{\net(f)}(\langle \mathbf{x}\rangle)$.

\label{thm:cpn_simulate_mpbn}
\end{theorem}

\begin{proof}
We first show the \enquote{$\Longrightarrow$} direction: the CPN retains all possible behaviours of the MPBN.

For any MP reachability, $\mathbf{x} \overset{mp}{\longrightarrow^{*}} \mathbf{y}$, there exists a witness $\mathbf{x}=\hat{\mathbf{x}}_0\xrightarrow{\hat{mp}} \hat{\mathbf{x}}_1 \xrightarrow{\hat{mp}} \dots\xrightarrow{\hat{mp}} \hat{\mathbf{x}}_k=\mathbf{y}$ of specific structure~\cite{PauleveKCH20}.
In particular, it can be split into three phases.
The first phase, $\mathbf{x}=\hat{\mathbf{x}}_0\xrightarrow{\hat{mp}} \dots\xrightarrow{\hat{mp}} \hat{\mathbf{x}}_h=\hat{\mathbf{z}}$, consists only of variables changing value from Boolean to transient.
In the second phase, $\hat{\mathbf{z}}=\hat{\mathbf{x}}_h\xrightarrow{\hat{mp}} \dots\xrightarrow{\hat{mp}} \hat{\mathbf{x}}_{h'}=\hat{\mathbf{z}}'$, all variables change value only between the two transient values and finally, in the third phase, $\hat{\mathbf{z}}'=\hat{\mathbf{x}}_{h'}\xrightarrow{\hat{mp}} \dots\xrightarrow{\hat{mp}} \hat{\mathbf{x}}_k=\mathbf{y}$, all variables collapse back into a Boolean value.
Additionally, each variable changes value at most once in each of the phases.

The first phase, $\mathbf{x}\xrightarrow{\hat{mp}} \dots\xrightarrow{\hat{mp}} \hat{\mathbf{z}}$, fits the conditions of Corollary~\ref{cor:abstraction_one_to_one}, immediately yielding $\langle  \mathbf{x}\rangle\rightarrow\dots\rightarrow\langle \hat{\mathbf{z}}\rangle$.
It remains to show
$\langle \mathbf{y} \rangle\in\limreach_{\net(f)}(\langle \hat{\mathbf{z}}\rangle)$.

We know that for any variable $i\in\{1,\dots, n\}$, $\hat{\mathbf{z}}_i \in \mathbb{B}$ implies $\hat{\mathbf{z}}_i = \mathbf{y}_i$, and thus ${\langle\hat{\mathbf{z}}\rangle}_i={\langle \mathbf{y}\rangle}_i$ and ${\langle\hat{\mathbf{z}}\rangle}_{i+n}={\langle \mathbf{y}\rangle}_{i+n}$.
There is therefore no need to fire transitions in $T_i$ and we only concern ourselves with variables having transient values in $\hat{\mathbf{z}}$.

We now show that for every such variable $i \in \{1,\dots, n\}$, $\hat{\mathbf{z}}_i \in \{\nearrow,\searrow\}$, there exists a transition $\transn\in T_i$ exactly $0.5$ enabled in $\hat{\mathbf{z}}$, and such that the marking $\langle \hat{\mathbf{z}}\rangle\xrightarrow{0.5\transn} \markn$ agrees with the target configuration $\mathbf{y}$ of the variable $i$, $\markn_i = 1 - \mathbf{y}_i$.

Let us first assume $\hat{\mathbf{z}}_i=\hat{\mathbf{z}}'_i$, i.e. the variable $i$ does not change value in the second phase.
Since $\hat{\mathbf{z}}_i \in \{\nearrow,\searrow\}$, we know there exists a unique transition $\transn\in T_i$ such that $\langle\hat{\mathbf{w}}\rangle\xrightarrow{0.5\transn}\langle\hat{\mathbf{w}}'\rangle$ occurs somewhere along the first phase, $\langle  \mathbf{x}\rangle\rightarrow\dots\rightarrow\langle \hat{\mathbf{z}}\rangle$.
For any variable $j\in \{1,\dots, n\}$, $\hat{\mathbf{z}}_j\in \mathbb{B}$ implies $\hat{\mathbf{z}}_j=\hat{\mathbf{w}}_j$.
We can thus apply Lemma~\ref{lem:enab_monotony} to obtain $\enab(\transn, \langle \hat{\mathbf{z}} \rangle) \geq 0.5$.
Finally, we know that by firing $\transn$ fully, the value of $i$ becomes flipped, $\mathbf{x}_i = \markn_i$.
$\hat{\mathbf{z}}_i=\hat{\mathbf{z}}'_i$ necessarily means $\mathbf{y}_i \neq \mathbf{x}_i$ and thus $\mathbf{x}_i = \markn_i = 1 - \mathbf{y}_i$.

Let now $\hat{\mathbf{z}}_i\neq\hat{\textbf{z}}'_i$ and let $\hat{\mathbf{v}} = \hat{\mathbf{x}}_g \xrightarrow{\hat{mp}} \hat{\mathbf{x}}_{g+1} = \hat{\mathbf{v}}'$, $h \leq g < h'$ be the configurations such that $\Delta(\hat{\textbf{v}},\hat{\mathbf{v}}') = \{i\}$.
By Lemma~\ref{lem:half_enabled}, there exists a transition $\transn\in T_i$ such that $\enab(\transn, \langle \hat{v}\rangle) \geq 0.5$.
Since variables are only changing value between transient ones along the second phase, we get $\langle \hat{\mathbf{v}}\rangle = \langle \hat{\mathbf{z}}\rangle$ and thus $\enab(\transn, \langle \hat{\mathbf{z}}\rangle) \geq 0.5$.
Let us now, without loss of generality, assume $\hat{\mathbf{v}}'_i = \nearrow$ and thus also $\transn\in T_i^+$.
Then we have $\markn_i = 0$ and since the variable $i$ does not change value again in phase two, only collapses to a Boolean value in phase three, $\mathbf{y}_i = 1$, or alternatively $1 - \mathbf{y}_i = 0 = \markn_i$.

We thus have a set of transitions $S \subseteq \trans$ such that no two transitions act on the same variable, $\forall i\in\{1,\dots, n\}, \transn\neq\transn'\in T_i \Longrightarrow \{\transn,\transn'\}\not\subseteq S$, and every variable with transient value in $\hat{\mathbf{z}}$ is represented, $\forall  i\in\{1,\dots, n\}, \hat{\mathbf{z}}_i\in{\{\nearrow,\searrow\}} \Longleftrightarrow T_i\cap S\neq \emptyset$.
Moreover, every transition in $S$ is $0.5$ enabled in $\langle\hat{\mathbf{z}}\rangle$, $\forall\transn\in S, \enab(\transn, \langle\hat{\mathbf{z}}\rangle)=0.5$.
Let thus $\sigma = (0.5\transn)_{\transn \in S}$ be a sequence of all such transitions firing in arbitrary order.
If $\sigma$ were a valid firing sequence, we would immediately get $\langle\hat{\mathbf{z}}\rangle\xrightarrow{\sigma}\langle \mathbf{y}\rangle$, however, $0.5$ firing a transition in $S$ might completely disable other transitions, as it empties a place.
Consider thus $\frac{\sigma}{2}$.
After firing any transition by $0.25$, each place is guaranteed to remain at least $0.25$ marked, thus guaranteeing that all other transitions are at least $0.25$ enabled. $\frac{\sigma}{2}$ is therefore a firing sequence.
We can repeat this reasoning to construct a discounted firing sequence which reaches $\langle \mathbf{y}\rangle$ in the limit, $\langle\hat{\mathbf{z}}\rangle\xrightarrow{(2^{-n}\sigma)_{n\in\mathbb{N}}}\langle \mathbf{y}\rangle$, concluding this proof direction.

We now show the \enquote{$\Longleftarrow$} direction: the CPN does not allow any more behaviours than the MPBN.

Let $\langle \mathbf{x}\rangle \xrightarrow[\net(f)]{\sigma} \langle \mathbf{y}\rangle$ be a (infinite) firing sequence witnessing the reachability of $\langle \mathbf{y}\rangle$ from $\langle\mathbf{x}\rangle$.
We can without loss of generality assume $\sigma$ includes no null firings.
We define a set of variables
\[
I \define\Delta(\mathbf{x}, \mathbf{y})\cup \left\{i\in\{1,\dots, n\}\mid\exists \langle\mathbf{x}\rangle \xrightarrow{\sigma_1} \markn \xrightarrow{\sigma_2} \langle \mathbf{y}\rangle, \sigma_1\sigma_2 = \sigma \wedge \markn_i \not\in \mathbb{B}\right\}
\]
as the variables whose value (marking of the corresponding places) changes at some point along $\sigma$.
We further impose a total order on $I$ according to the associated places of which variable have their value modified first, resulting in the vector $\overrightarrow{\bf{I}} = (j_0,\dots, j_k)$.
The order is total as every transition modifies only two places $i, i+n$, which are associated to the same variable $i\in\{1, \dots, n\}$.

We now show that the MPBN $f$ can reach a configuration $\hat{\mathbf{z}}$ such that for every $i\in I$, $\hat{\mathbf{z}}_i\in\{\nearrow,\searrow\}$.
We conduct the proof by induction on $k$.

In the base case, we have to show $\mathbf{x} \xrightarrow{\hat{mp}} \hat{\mathbf{w}}$ where $\hat{\mathbf{w}}_{j_0} \in \{\nearrow,\searrow\}$ and $\forall i\neq j_0, \hat{\mathbf{w}}_i\in\mathbb{B}$.
This follows directly from Lemma~\ref{lem:enab_transfers_to_bn}, as $\langle \mathbf{x}\rangle \in \mu(\mathbf{x})$.

Let now $h< k$ be such that $\mathbf{x}\overset{\hat{mp}}{\longrightarrow^*} \hat{\mathbf{w}}$ and $\forall i\in \{1,\dots,n\}$, $\hat{\mathbf{x}}_i\in\mathbb{B}\Longleftrightarrow i\not\in\{j_0,\dots,j_h\}$.
Let then $\langle \mathbf{x}\rangle \xrightarrow{\sigma_1} \markn \xrightarrow{\alpha\transn} \markn'\xrightarrow{\sigma_2} \langle \mathbf{y}\rangle$ be such that $\sigma_1(\alpha\transn)\sigma_2=\sigma$, $\support{\sigma_1}\subseteq\bigcup_{i=0}^h T_i$ and $\transn\in T_{h+1}$.
Such markings and transition have to exist as per definition of $\overrightarrow{\bf{I}}$.

Let further $\hat{\mathbf{v}}\in\{0,1,\nearrow,\searrow\}^n$ be arbitrary such that $\markn\in\mu(\hat{\mathbf{v}})$.
We know such a configuration exists as the P-invariants of $\net(f)$ guarantee that $\markn_i + \markn_{i+n} = 1$ is preserved for any variable $i\in\{1,\dots, n\}$.
By Lemma~\ref{lem:mark_equivalence}, we thus immediately have $\enab(\transn,\langle\hat{\mathbf{v}}\rangle) \geq 0.5$.
$\support{\sigma_1}\subseteq\bigcup_{i=0}^h T_i$ guarantees that for any variable $i\in\{1,\dots,n\}$, $\hat{\mathbf{w}}_i\in\mathbb{B} \Longrightarrow \hat{\mathbf{v}}_i=\hat{\mathbf{w}}_i$, as all such variables are untouched since the initial configuration, allowing us to apply Lemma~\ref{lem:enab_monotony} to obtain $\enab(\transn,\langle\hat{\mathbf{w}}\rangle) \geq 0.5$.
Finally, by Lemma~\ref{lem:enab_transfers_to_bn} we get $\hat{\mathbf{w}}\xrightarrow{\hat{mp}}\hat{\mathbf{w}}'$ with $\hat{\mathbf{w}}'_{j_{h+1}}\in\{\nearrow,\searrow\}$.

We thus have $\mathbf{x}\overset{\hat{mp}}{\longrightarrow^*} \hat{\mathbf{z}}$ such that $\hat{\mathbf{z}}_i\in\{\nearrow,\searrow\}\Longleftrightarrow i\in I$.
Additionally, since $\Delta(\mathbf{x},\mathbf{y})\subseteq I$, we also know $\hat{\mathbf{z}}_i \in \mathbb{B} \Longleftrightarrow \hat{\mathbf{z}}_i=\mathbf{y}_i$.
Let us now consider the variables $i\in\{1,\dots, n\}$ such that $\hat{\mathbf{z}}_i = \searrow$ and $\mathbf{y}_i = 1$, respectively $\hat{\mathbf{z}}_i = \nearrow$ and $\mathbf{y}_i = 0$.
As the two cases are symmetrical, we assume $\mathbf{y}_i = 1$ without loss of generality.

Since $i\in I$ we know that there exists $\langle \mathbf{x}\rangle\xrightarrow{\sigma_1}\markn\xrightarrow{\alpha\transn}\markn'\xrightarrow{\sigma_2} \langle \mathbf{y}\rangle$ such that $\sigma_1(\alpha\transn)\sigma_2=\sigma$ and $\markn_i > \markn'_i$.
We apply the same reasoning as in the above induction step. Let thus $\hat{\mathbf{v}}\in\{0,1,\nearrow,\searrow\}^n$ be arbitrary such that $\markn\in\mu(\hat{\mathbf{v}})$.
By Lemma~\ref{lem:mark_equivalence}, $\enab(\transn, \langle\hat{\mathbf{v}}\rangle) \geq 0.5$.
$\forall i\in I, \hat{\mathbf{z}}_i \in\{\nearrow,\searrow\}$ gives us for any $j\in\{1,\dots,n\}$, $\hat{\mathbf{z}}_j\in\mathbb{B}\Longrightarrow \hat{\mathbf{z}}_j=\hat{\mathbf{v}}_j$, allowing us to apply Lemma~\ref{lem:enab_monotony} to get $\enab(\transn, \langle\hat{\mathbf{z}}\rangle) \geq 0.5$.
And finally, by Lemma~\ref{lem:enab_transfers_to_bn}, $\hat{\mathbf{z}} \xrightarrow{\hat{mp}} \hat{\mathbf{w}}$ where $\hat{\mathbf{w}}_i = \nearrow$.

Since $\langle\hat{\mathbf{z}}\rangle=\langle\hat{\mathbf{w}}\rangle$, we can repeat the above for every such variable until we reach $\hat{\mathbf{z}}'$ such that for each variable $i\in\{1,\dots,n\}$, either $\hat{\mathbf{z}}'_i\in\mathbb{B}$ and $\hat{\mathbf{z}}'_i=\mathbf{y}_i$ or $\hat{\mathbf{z}}'_i\in\{\nearrow,\searrow\}$ and $\hat{\mathbf{z}}'_i=\nearrow \Longleftrightarrow \mathbf{y}_i=1$.
All that's left is thus to collapse all the variables into the corresponding Boolean values, giving us the coveted $\mathbf{x}\overset{\hat{mp}}{\longrightarrow^*}\hat{\mathbf{z}}\overset{\hat{mp}}{\longrightarrow^*}\hat{\mathbf{z}}'\overset{\hat{mp}}{\longrightarrow^*} \mathbf{y}$.
\end{proof}

By Theorem~\ref{thm:cpn_simulate_mpbn}, the net from Figure~\ref{fig:pn_encoding} instantiated as a CPN with the initial marking $\markn_0=(1,1,1,0,0,0)=\langle(0,0,0)\rangle$ recovers exactly the reachability graph given in Figure~\ref{fig:mp_ts}.
We turn to the symbolic analysis methods to visualize this.
We omit the ARG due to scale (the graph has $27$ vertices) and instead show an excerpt of the SRT in Figure~\ref{fig:re_srt}.

\begin{figure}
    \centering
    \includegraphics[width=\textwidth]{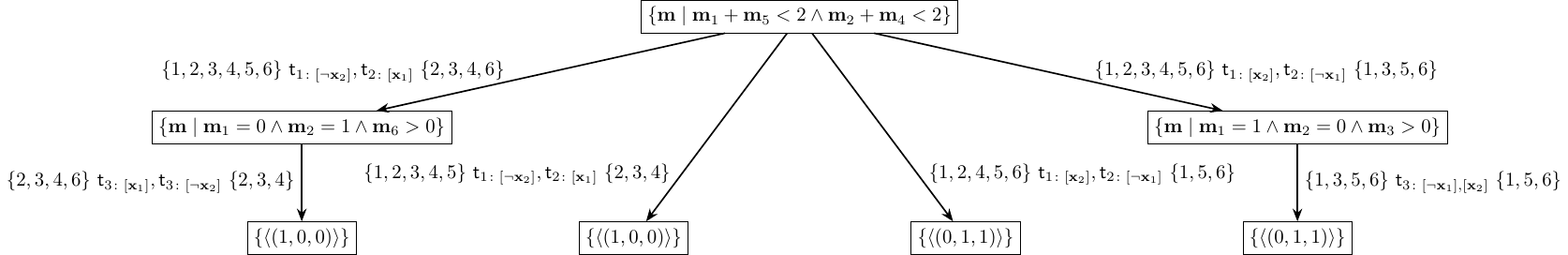}
    \caption{An excerpt of the SRT for the CPN encoding of the example BN (Figure~\ref{fig:pn_encoding}). The remaining $14$ branches, each isomorphic to one of the listed branches with respect to vertex labels but with smaller supports, have been pruned for readability.}
    \label{fig:re_srt}
\end{figure}

The SRT is a symbolic representation of the reachable configurations in Figure~\ref{fig:mp_ts}.
All configurations $\mathbf{x}\in\mathbb{B}^3$ such that $\mathbf{x}_1 = \mathbf{x}_2$ can update to any other configuration and their associated markings, $\langle\mathbf{x}\rangle$, belong to the root of the SRT.

The leaves of the SRT are the fixed points of the MPBN.
The same fixed point appears in multiple leaves due to different ways the system can reach it.
The two arrows in the middle reach the attractors directly, in case the variable $3$ already has the \enquote{correct} value.
The branches on the far left and far right then cross over an intermediate vertex.
On the left, this vertex represents the configuration $(1,0,1)$ and on the right $(0,1,0)$.

As seen in Figure~\ref{fig:mp_ts}, the configurations $(1,0,1)$ and $(0,1,0)$ constitute the strong basins of attraction of the respective fixed points.
The SRT thus captures not only the attractors themselves, but also the system \enquote{deciding} on the long-term behaviour by leaving the weak basin of attraction of some attractor.

\section{Conclusion and Outlook}
We have demonstrated that the PN encoding of BNs which has been previously introduced to simulate fully asynchronous semantics as a DPN, can equally be used to simulate the much richer, MP semantics as a CPN.
The proof is educational, showing us that finite reachability is not enough to capture simultaneous variable updates and lim-reachability is required to reproduce all behaviours exhibited by the MPBN.

The ability to capture MPBNs by CPNs naturally allows MPBNs to benefit from all results available for CPNs.
Similarly to MPBNs, CPN reachability is polynomial, a result which is not limited to locally monotonic Boolean functions.
We further exhibited that tools for symbolic analysis of CPNs recapitulate the reachable configurations of MPBNs with special focus on \enquote{key events} which shape the attractor landscape.

We believe that the presented results serve both to open up new avenues for studying and understanding MPBNs, and motivate the further study of the \enquote{abstraction scale} of dynamical models of interaction.

\newpage
\bibliographystyle{abbrv}
\bibliography{references}

\end{document}